\documentclass[aps,prl,twocolumn,superscriptaddress,showpacs,preprintnumbers]{revtex4-1}
\usepackage{mathrsfs}
\usepackage{braket}
\usepackage{graphicx}
\usepackage{dcolumn}
\usepackage{amsmath}
\usepackage{amsfonts}
\usepackage{epsfig}
\usepackage{subfigure}
\usepackage{color}
\usepackage{epstopdf}
\usepackage{framed}
\usepackage[english]{babel}
\usepackage[utf8]{inputenc}

\newtheorem{theorem}{Theorem}

\newtheorem{corollary}[theorem]{Corollary}

\newenvironment{proof}[1][Proof]{\noindent\textbf{#1.} }{\ \rule{0.5em}{0.5em}}

\newcommand\argmin{\mathop{\mathrm{argmin}}}

\def\be{\begin{equation}}
\def\ee{\end{equation}}
\def\ba{\begin{eqnarray}}
\def\ea{\end{eqnarray}}

\begin{document}
\date{\today}
\title{
Additive Classical Capacity of Quantum Channels Assisted by Noisy Entanglement
}

\author{Quntao Zhuang}
\email{quntao@mit.edu}
\affiliation{
Department of Physics, 
Massachusetts Institute of Technology, Cambridge, MA 02139, USA}
\affiliation{Research Laboratory of Electronics, 
Massachusetts Institute of Technology, Cambridge, MA 02139, USA}
\author{Elton Yechao Zhu}
\affiliation{
Department of Physics, 
Massachusetts Institute of Technology, Cambridge, MA 02139, USA}
\affiliation{Center For Theoretical Physics,
Massachusetts Institute of Technology, Cambridge, MA 02139, USA}
\author{Peter W. Shor}
\affiliation{Center For Theoretical Physics,
Massachusetts Institute of Technology, Cambridge, MA 02139, USA}
\affiliation{
Department of Mathematics, 
Massachusetts Institute of Technology, Cambridge, MA 02139, USA}

\date{\today}

\begin{abstract}
We give a capacity formula for the classical information transmission over a noisy quantum channel, with separable encoding by the sender and limited resources provided by the receiver's pre-shared ancilla. Instead of a pure state, we consider the signal-ancilla pair in a mixed state, purified by a ``witness''. Thus, the signal-witness correlation limits the resource available from the signal-ancilla correlation. Our formula characterizes the utility of different forms of resources, including noisy or limited entanglement assistance, for classical communication. With separable encoding, the sender's signals across multiple channel uses are still allowed to be entangled, yet our capacity formula is additive. In particular, for generalized covariant channels our capacity formula has a simple closed-form. Moreover, our additive capacity formula upper bounds the general coherent attack's information gain in various two-way quantum key distribution protocols. For Gaussian protocols, the additivity of the formula indicates that the collective Gaussian attack is the most powerful.

\end{abstract}
\maketitle

Communication channels model the physical medium for information transmission between the sender (Alice) and the receiver (Bob). Classical information theory~\cite{Shannon_1948,CT_book} says that a channel is essentially characterized by a single quantity---the (classical) channel capacity, {\it i.e.} its maximum (classical) information transmission rate. However, quantum channels~\cite{RMP_channel} can transmit information beyond classical. Formally, a (memoryless) quantum channel is a time-invariant completely positive trace preserving (CPTP) linear map between quantum states. Various types of information lead to various capacities, \textit{e.g.}, classical capacity $\mathcal{C}$~\cite{HSW1,HSW2} for classical information transmission encoded in quantum states and quantum capacity $\mathcal{Q}$~\cite{quantum_capacity_Lloyd,quantum_capacity_Shor,quantum_capacity_Devetak} for quantum information transmission. For both cases, implicit constraints on the input Hilbert space, \textit{e.g.}, fixed dimension or energy, quantify the resources. Resources can also be in the form of assistance: given unlimited entanglement, one has the entanglement-assisted classical capacity $\mathcal{C}_E$~\cite{Bennett_2002}. Ref.~\onlinecite{Wilde2012,Wilde2012_2} provide a capacity formula for the trade-off of classical and quantum information transmission and entanglement generation (or consumption).

With the trade-off capacity formula in hand, it appears that the picture of communication over quantum channels is complete. However, our understanding about the trade-off is plagued by the ``non-additivity'' issue~\cite{RMP_channel}, best illustrated by the example of $\mathcal{C}$. The Holevo-Schumacher-Westmoreland (HSW) theorem~\cite{HSW1,HSW2} gives the one-shot capacity $\mathcal{C}^{\left(1\right)}\left(\Psi\right)$ of channel $\Psi$, which assumes product-state input in multiple channel uses. Consider the tensor product channel $\Psi^{\otimes M}$, it may have one-shot capacity $\mathcal{C}^{\left(1\right)}\left(\Psi^{\otimes M}\right)> M\mathcal{C}^{\left(1\right)}\left(\Psi\right)$, since it allows the input state of $\Psi^M$ to be entangled across $M$ channel uses of $\Psi$ ($M-$shot). $\mathcal{C}\left(\Psi\right)$ is then given by the regularized expression as $\lim_{M\to \infty}\mathcal{C}^{\left(1\right)}\left(\Psi^{\otimes M}\right)/M$, which is difficult to calculate since the dimension of the input states of $\Psi^{\otimes M}$ is exponential in $M$. If we have the additivity property $\mathcal{C}^{\left(1\right)}\left(\Psi^{\otimes M}\right)=M\mathcal{C}^{\left(1\right)}\left(\Psi\right)$, the formula of the capacity is greatly simplified, \textit{i.e.} $\mathcal{C}\left(\Psi\right)=\mathcal{C}^{\left(1\right)}\left(\Psi\right)$. However, both $\mathcal{C}$~\cite{Hastings_2009} and $\mathcal{Q}$~\cite{Dur_2004} are known to be non-additive. Without additivity, quantification of the trade-off is in general infeasible.

An exception is the (unlimited) entanglement-assisted classical capacity $\mathcal{C}_E$~\cite{Bennett_2002}. 
Since it has the form of quantum mutual information~\cite{Wilde_book,Chuang_book}, $\mathcal{C}_E$ is additive~\cite{Bennett_2002,Adami_1997}.
One immediately hopes that the additivity can be extended to classical communication assisted by imperfect entanglement, since entanglement is fragile. Many such scenarios have been explored, {e.g.} superdense coding (SC) over a noisy channel assisted by noisy entanglement~\cite{dense_coding_1,dense_coding_2,dense_coding_3,dense_coding_4,dense_coding_5,dense_coding_6}, noiseless channel assisted by noisy entanglement~\cite{Horodecki01} and noisy channels assisted by limited pure state entanglement~\cite{Shor04}. However, all results are in general non-additive as expected~\cite{Zhu_2017}, since the above imperfect scenarios include the case with zero entanglement assistance---the non-additive $\mathcal{C}$.

In this paper, we obtain an additive classical capacity formula for a noisy quantum channel $\Psi$ assisted by resources such as noisy entanglement. In the most general formalism, Alice sends an optimized ensemble of (possibly mixed) states $\rho^i_{SE}$ to Bob, with signal $S$ through the channel $\Psi$ and an ancilla $E$ pre-shared through the identity channel $\mathcal{I}$. Each $\rho^i_{SE}$ is constrained by some resource, {\it e.g.} by the entanglement between $S$ and $E$. Here, similar to SC, we consider a restricted scenario of two-step signal preparation---resource distribution and encoding (see Fig.~\ref{scheme_single}). Each $\rho^i_{SE}$ is obtained by encoding on $S$ from a certain state $\rho_{SE}$. Moreover, the resource is constrained by the correlation between $S$ and a ``witness'' $W$---a purification of $\left(S,E\right)$.

In the resource distribution step, $W$ is made inaccessible to both Alice and Bob.
Instead of explicitly quantifying the available resource (between $S,E$) as in Ref.~\onlinecite{Shor04}, we describe the resource implicitly by quantifying the correlation between $S$ and $W$---the unavailable resource---by $K\ge1$ inequalities
\be
Q_k\left(\rho_{SW}\right)\ge y_k, k\in[1,K]
\label{singleform}
\ee 
on $\rho_{SW}$, where each $Q_k\left(\cdot\right)$ is a function on bipartite states. We denote Eqn.~\ref{singleform} by ${\bf Q}\left(\rho_{SW}\right)\ge  {\bf y}$. While Ref.~\onlinecite{Bennett_2002,Shor04} only considered pure state entanglement, the form of resources in our case can be arbitrary by choosing different $Q_k\left(\cdot\right)$, \textit{e.g.}, noisy entanglement, cross correlation~\cite{Gaussian_monogamy,Quntao_2015,Weedbrook_2012} or quantum discord~\cite{Discord}. However, entanglement measures are more meaningful to consider because: (1) they respect the unitary equivalence of the purification $W$; (2) constraints on the entanglement between $S$ and $W$ leads to constraints on the entanglement between $S$ and $ E$---a property known as monogamy~\cite{CKW_conjecture,Osborne_2006,Monogamy_tangle,Gaussian_monogamy}. 

Here we give an example of Eqn.~\ref{singleform}---the quantum mutual information~\cite{Wilde_book,Chuang_book} $I\left(S:W\right)\ge y$, $y\in[0, 2\log_2d]$ for qudit $S$. When $y=2\log_2d$, $\rho_{SW}$ is pure and thus $E$ and $S$ are uncorrelated. Since entanglement across multiple channel uses is also excluded here, the additivity of our capacity does not contradict the non-additivity of $\mathcal{C}$. When $y=0$, the optimum has $W$ and $S$ in a product state and $\rho_{SE}$ pure as in Ref.~\cite{Shor04}. This gives the case of Ref.~\cite{Bennett_2002}. For intermediate values of $y$, $\rho_{SE}$ is mixed and signals across multiple channel uses can be entangled, thus the additivity of our capacity is non-trivial. This example illustrates the desired property of function $Q_k\left(\cdot\right)$---the correlation between $S,W$ increases when $Q_k\left(\cdot\right)$ increases, with the two end points corresponding to $\rho_{SW}$ pure and product state.

In the encoding step, Alice performs a quantum operation $\varepsilon_x$~\cite{locc_explain} with probability $P_X\left(x\right)$ on $S$ to encode a message $x$, resulting in $S^\prime$ as the input to $\Psi$. In multiple channel uses, the encoding is a set of classically correlated separate operations---local operations and classical communication (LOCC)~\cite{LOCC}. $\rho_S$ is constrained to be in $\mathcal{B}\left(\mathcal{H}_S\right)$---density operators on Hilbert space $\mathcal{H}_S$, and the encoding is constrained to be in a certain set, \textit{i.e.}, $\left(P_X\left(\cdot\right), ~\varepsilon_\cdot\right)\in \mathbb{G}$. Upon receiving $\Psi's$ output $B$, Bob makes a joint measurement on $B$ and $E$ to determine $x$. The capacity of the above scenario is given as follows.
\begin{theorem}
{\rm (Classical capacity with limited resources and LOCC encoding.)}
With resources constrained by $V\equiv \left\{\,\left(P_X\left(\cdot\right), ~\varepsilon_\cdot\right)\in \mathbb{G},~\rho_{S}\in \mathcal{B}\left(\mathcal{H}_S\right),~{\bf Q}\left(\rho_{SW}\right)\ge{\bf y}\,\right\}$, suppose $\mathbb{G}$ allows arbitrary phase flips, the classical capacity of the quantum channel $\Psi$ is
\begin{align}
&\chi_L\left(\Psi\right)=\max_V S\left(\sum_x P_X\left(x\right) \Psi\circ \varepsilon_x \left[\rho_{S}\right]\right)&
\nonumber\\
&
-\sum_x P_X\left(x\right) E_{\Phi_{\mbox{$\varepsilon_x$}}\otimes \mathcal{I}}\left[\rho_{SW}\right],
\label{capacity_full}
\end{align}
where $\Phi_{\mbox{$\varepsilon_x$}}$ is the complementary quantum operation to $\Psi\circ\varepsilon_x$, the entropy gain $E_\phi$~\cite{entropy_gain} of a CPTP map $\phi$ on state $\rho$ is defined by
$
E_\phi \left[\rho\right]\equiv S\left(\phi\left[\rho\right]\right)-S\left(\rho\right),
$
and the maximization is over the encoding $\left(P_X\left(\cdot\right), ~\varepsilon_\cdot\right)$ and $\rho_{SW}$.
Eqn.~\ref{capacity_full} is additive when the constraint has a separable form on each channel use and the encoding is LOCC.
\label{channel_capacity_theorem}
\end{theorem}
We make two clarifications about the theorem. First, a schematic of $\Phi_{\mbox{$\varepsilon_x$}}$ is given in Fig.~\ref{scheme_single}. The encoding CPTP map $\varepsilon_x$ is extended to a unitary operation $U_x$ on $S$ and an environment $C$ in the vacuum state, resulting in $S^\prime$ in state $\varepsilon_x\left[ \rho_{S}\right]$ and $C^\prime$. $S^\prime$ is sent to Bob through $\Psi$, whose Stinespring's dilation is a unitary operation $U_\Psi$ on $S^\prime$ and an environment $N$ in the vacuum state, producing $B$ for Bob and an environment $N^\prime$. We define $\Phi_{\mbox{$\varepsilon_x$}}$ as the CPTP map from $\rho_{S}$
to $\rho_{N^\prime C^\prime}^{\left(x\right)}$, given $\varepsilon_x$.
Second, by a separable form of constraints on each channel use, we mean constraints expressed by a set of inequalities, each involving states only in a single channel use (see Eqn.~\ref{Mform2}).

\begin{figure}
\includegraphics[width=0.35\textwidth]
{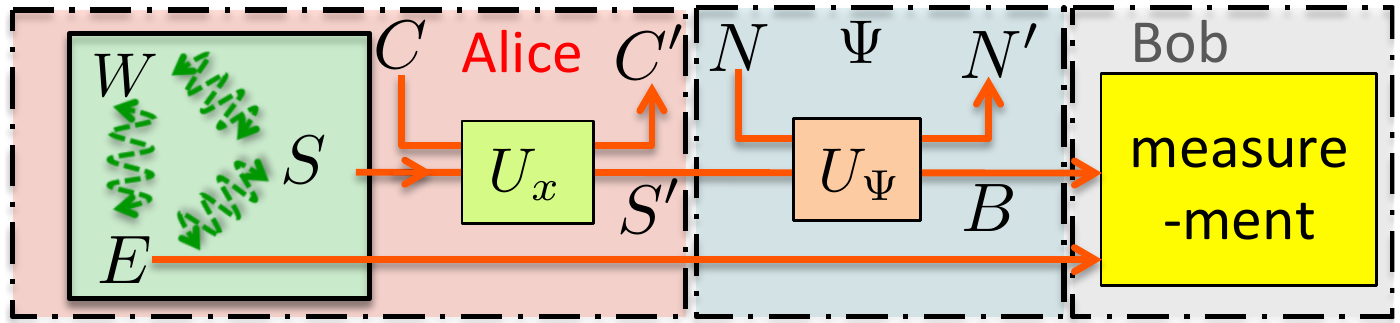}
\caption{Schematic of a single channel use.}
\label{scheme_single}
\includegraphics[width=0.35\textwidth]
{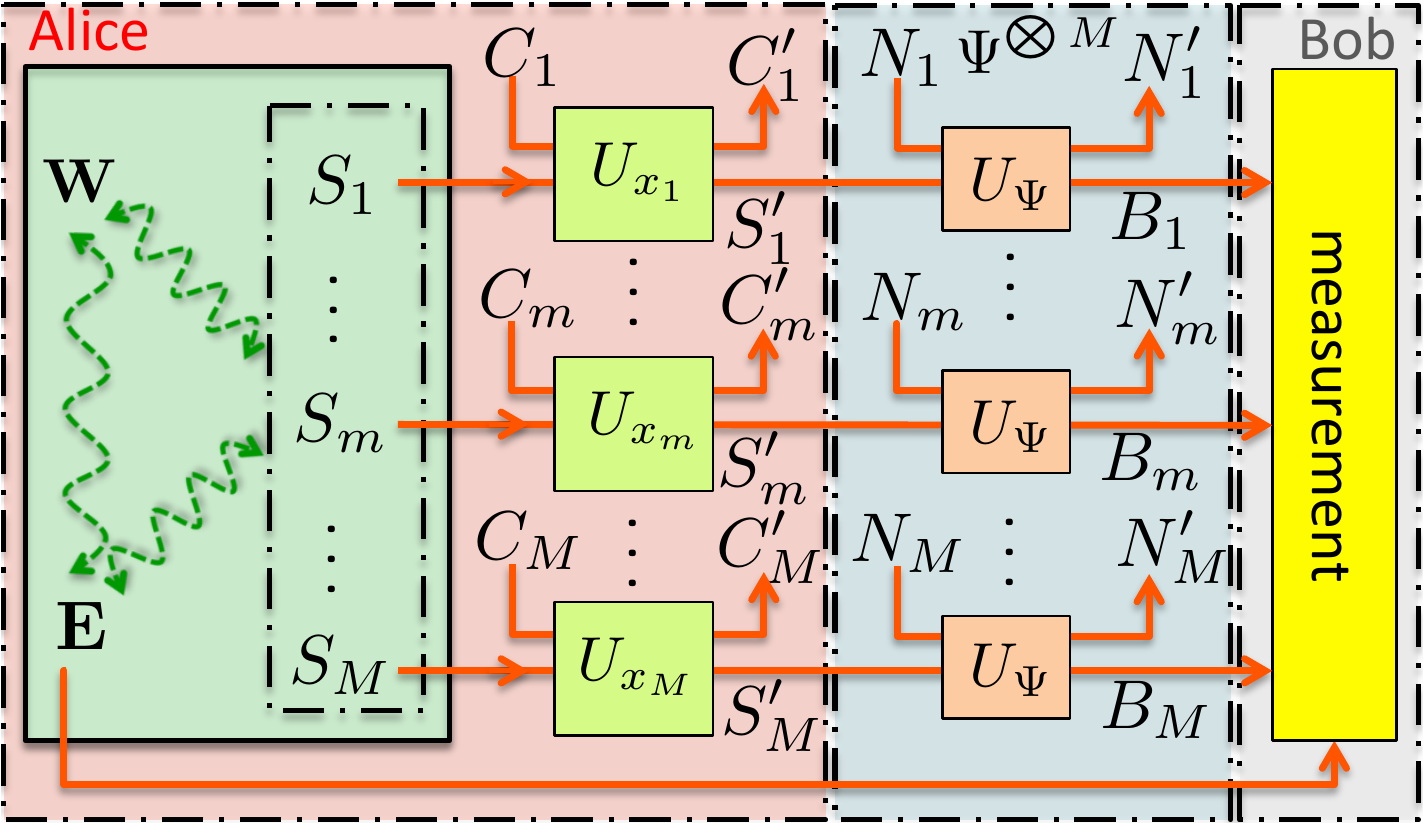}
\caption{Schematic of $M$ channel uses.}
\label{scheme_general} 
\end{figure}

We have given our main result ``theorem \ref{channel_capacity_theorem}'' in a single channel use scenario. In order to prove additivity, we need to consider multiple channel uses (Fig.~\ref{scheme_general}). 
Before that, we make a few more comments. First, for generalized covariant channels, including covariant~\cite{covariant} channels and Weyl-covariant~\cite{covariant_Weyl} channels, Eqn.~\ref{capacity_full} can be simplified. More details are given in corollary~\ref{channel_capacity_theorem_covariant_finite}.

Next, we discuss the relationships with other capacities. If $\mathbb{G}$ allows arbitrary encoding, one can choose to replace the original signal state with an optimal set of pure states, which guarantees that $\chi_L\ge \mathcal{C}^{(1)}$. With all encoding operations unitary, we obtain another lower bound $\chi_L^{\mathcal{I}}$. When $y_k$'s are maximum, $\chi_L=\mathcal{C}^{(1)}$; when $y_k$'s are minimum, $\chi_L=\mathcal{C}_E$;  Note when arbitrary phase flips are not allowed, the r.h.s. of Eqn.~\ref{capacity_full} upper bounds $\chi_L$, and it is still additive while $\chi_L$ might not be. We also point out that Ref.~\onlinecite{Shor04} and our result are different in the sense that neither of them can be reduced to the other. If $\varepsilon_x$'s are not unitary, then the environment $C^\prime$ is never sent to Bob. This is different from Ref.~\onlinecite{Shor04}, where all purification of the signal is sent to Bob. If we restrict $\varepsilon_x$'s to be unitary, the input states in Ref.~\onlinecite{Shor04} do not need to be related by unitary operations, different from our scenario~\cite{note2}.

Finally, we emphasize the application of our results. Our capacity formula provides an additive upper bound for the general eavesdropper's coherent attack~\cite{Renner_2008,Valerio_2009, Furrer_2012,Leverrier_2013} information gain for various two-way quantum key distribution (TW-QKD) protocols~\cite{Quntao_2015,Ping_Pong,two_way_no_loss,generalization_no_loss,modified_Ping_Pong,Pirandola_2008,Zhang_2014,Zheshen_2016,Ottaviani_2014,Ottaviani_2015,Ottaviani_2016}. The constraint in Eqn.~\ref{singleform} appears in security checking of TW-QKD protocols, where two parties verify properties of their state $\rho_{SW}$ to constrain the eavesdropper's benefit from $\left(S,E\right)$ (details in corollary~\ref{capacity_Eve}). 
Obtaining upper bounds for eavesdroppers in TW-QKD is more complicated than for one-way protocols due to the simultaneous attack on both the forward and the backward channels. Only special attacks~\cite{modified_Ping_Pong,Pirandola_2008,Zhang_2014,Ottaviani_2014,Ottaviani_2015,Ottaviani_2016} or general attacks in the absence of loss and noise~\cite{Ping_Pong,two_way_no_loss,generalization_no_loss} have been considered. Despite this difficulty, a TW-QKD protocol called ``Floodlight QKD'' has recently been shown to have the potential of reaching unprecedented secret key rate (SKR)~\cite{Quntao_2015,Zheshen_2016}. Consequently, our upper bound is crucial for high-SKR QKD.

{\em Multiple channel uses.---}Now we extend the single channel use scenario to $M\ge2$ channel uses in a non-trivial way that allows an additive classical capacity (Fig.~\ref{scheme_general}). 
We keep the same notation for all the modes except for adding a subscript to index the channel use. For convenience, we introduce the short notation ${\bf S} =\left\{\,S_m : m\in[1,M]\,\right\}$ for input signals, with its states $\rho_{\bf S}\in \mathcal{B}\left(\mathcal{H}_S^{\otimes M}\right)$, and also $\bf W$ for arbitrary inaccessible witness and $\bf E$ for arbitrary ancilla. Then the initial state $\left(\bf S, \bf E, \bf W\right)$ is pure.

The allowed encoding operations in $M$ channel uses are LOCC, \textit{i.e.}, they can be classically correlated, satisfying some joint distribution $P_{\bf X}\left(\cdot\right)$, where ${\bf X}=\left(X_1,\cdots, X_M\right)$ denotes the symbols in $M$ channel uses. Conditioned on the message ${\bf x}\equiv\left(x_1,\cdots, x_M\right)$, the encoding operation is $\varepsilon_{\bf x}= \otimes_{m=1}^M \varepsilon_{x_m}$.
Again the CPTP map $\varepsilon_{\bf x}$ can be extended as a unitary operation $\otimes_{m=1}^M U_{x_m}$, which takes in the signals $\bf S$ and the environment ${\bf C}=\left\{\,C_m : m\in[1,M]\,\right\}$ in the vacuum state and produces the encoded signals ${\bf S'}=\left\{\,S_m' : m\in[1,M]\,\right\}$ and environment ${\bf C^\prime}=\left\{\,C_m^\prime : m\in[1,M]\,\right\}$. Each encoding operation $\varepsilon_{x_m}$ with its own marginal distribution $P_{X_m}\left(\cdot\right)$ is still constrained to be inside the same set $\mathbb{G}$.

After the encoding step, each $S_m'$ is sent through $\Psi$ separately. The Stinespring's dilation of $\Psi^{\otimes M}$ takes $\bf S'$ and an environment ${\bf N}=\left\{\,N_m : m\in[1,M]\,\right\}$ in the vacuum state as inputs and outputs ${\bf B}=\left\{\,B_m : m\in[1,M]\,\right\}$ for Bob and the environment ${\bf N^\prime}=\left\{\,N_m^\prime : m\in[1,M]\,\right\}$. Bob decodes the message by joint measurements on $\left(\bf B,\bf E\right)$, where the pre-shared ancilla $\bf E$ provides resources quantified by the constraint
$
{\bf Q}\left(\rho_{S_m\bf W}\right)\ge  {\bf y}, m\in[1,M].
$
One can also consider $M$ witnesses ${\bf W}=\{\,W_m : m\in[1,M]\,\}$, with constraints on each signal-witness pair,
\be
{\bf Q}\left(\rho_{S_mW_m}\right)\ge  {\bf y}, m\in[1,M].
\label{Mform2}
\ee
Note that both constraints have a separate form on each channel use, allow entanglement between $S_m$'s across channel uses when ${\bf y}$ is not maximum and give the same additive capacity formula in theorem~\ref{channel_capacity_theorem}~\cite{note2}.

{\em Proof of theorem~\ref{channel_capacity_theorem}.---}
With the $M$-channel-use scenario established, we now prove theorem~\ref{channel_capacity_theorem}. The one-shot classical capacity of the product channel $\Psi \otimes \mathcal{I}$ for $\left(S',E\right)$ is given by the constrained version of the HSW theorem
\be
\label{HSW}
\chi_L\left(\Psi\right)=\max_V
 \left\{\,
S\left(\rho_{BE}\right) -\sum_{x}P_X\left(x\right)S\left(\rho_{BE}^{\left(x\right)}\right)
\,\right\},
\ee
where the maximization is over the encoding $\left(P_X\left(\cdot\right), \varepsilon_\cdot\right)$ and the source $\rho_{SW}$ constrained by $V$, and ${\rho_{BE}^{\left(x\right)}=\left(\Psi\circ \varepsilon_x\right)\otimes \mathcal{I} \left[\rho_{SE}\right]}$, with ${\rho_{BE}=\sum_x P_X\left(x\right) \rho_{BE}^{\left(x\right)}}$. Because $\left(S, E, W\right)$ and $N$, $C$ are pure, $S\left(\rho_{E}\right)=S\left(\rho_{SW}\right)$; it also follows that $\left(B, E, W, N^\prime,C^\prime\right)$ is pure, conditioned on $x$. Thus $S\left(\rho_{B E}^{\left(x\right)}\right)=S\left(\rho_{N^\prime C^\prime W}^{\left(x\right)}\right)$. Using the sub-additivity of von Neumann entropy on $S\left(\rho_{BE}\right)$ and combining the above equalities,
\begin{align}
& \chi_L\left(\Psi\right)\le\chi^{\rm UB}_L\left(\Psi\right)\equiv\max_V
 \biggl\{\,
S\left(\rho_{B}\right)
&
\nonumber\\
&
-\sum_x P_X\left(x\right)\left[S\left(\rho_{N^\prime C^\prime W}^{\left(x\right)}\right)-S\left(\rho_{SW}\right)\right]\,\biggr\}.&
\label{upperbounddef}
\end{align}

Noticing that $\Phi_{\mbox{$\varepsilon_x$}}$ maps $S$ to $N^\prime C^\prime$, Eqn.~\ref{upperbounddef} can be expressed as $\chi_L^{\rm UB}\left(\Psi\right)=\max_V F\left[\rho_{SW},\left(\,P_X\left(\cdot\right),\varepsilon_\cdot\,\right)\right]$, where
\be
F\left[\rho_{SW},\left(\,P_X\left(\cdot\right),\varepsilon_\cdot\,\right)\right]\equiv 
S\left(\rho_{B}\right)-\sum_x P_X\left(x\right) E_{\Phi_{\mbox{$\varepsilon_x$}}\otimes \mathcal{I}}\left[\rho_{SW}\right].
\label{upperbound1}
\ee
It's subadditive since $E_\phi$ is superadditive~\cite{note2}.

Now we switch to the $M$ channel uses scenario to prove additivity. If we adopt constraint~\ref{Mform2}, the overall constraint $V^{\left(M\right)}$ is in a separable form of $\left\{\,V_m, m\in[1,M]\,\right\}$, where
$
V_m\equiv \left\{\,\left(P_{X_m}\left(\cdot\right), \varepsilon_\cdot\right)\in \mathbb{G},~\rho_{S_m}\in \mathcal{B}\left(\mathcal{H}_S\right),~
{\bf Q}\left(\rho_{S_mW_m}\right)\ge  {\bf y}\,\right\}
$.
This separable form and the LOCC encoding allows the upper bound~\cite{note2}
$
\chi_L^{\rm UB}\left(\Psi^{\otimes M}\right) \le\sum_{m=1}^M \max_{V_m}F\left[\rho_{S_mW_m},\left(\,P_{X_m}\left(\cdot\right),\varepsilon_\cdot\,\right)\right],
$
which can be achieved~\cite{note2} by block encoding~\cite{Shor04}, leading to Eqn.~\ref{capacity_full} since $\rho_{B}=\sum_x P_X\left(x\right) \Psi\circ \varepsilon_x \left[\rho_{S}\right]$.

{\em Special case: generalized covariant channels.---} 
Consider a $d$-dimensional channel $\Psi$, we define its covariant group $G\left(\Psi\right):=\{U\in U(d):\,\forall\,\mbox{density matrix}\,\rho,\,\exists\,V\in U(d), \,s.t.\,\Psi\left(U\rho U^\dagger\right)=V\Psi\left(\rho \right)V^\dagger\} $, where $U(d)$ is the $d$ dimensional unitary group. If there exists a subset $G_U\left(\Psi\right)\subset G\left(\Psi\right)$ of size $d^2$ such that $\sum_{U_x\in G_U\left(\Psi\right)} U_x M U_x^\dagger=0$ for all $d\times d$ traceless matrices $M$~\cite{Horodecki01}, we call $\Psi$ {\it generalized covariant}. Generalized covariant channels include covariant channels~\cite{covariant} and Weyl-covariant channels~\cite{covariant_Weyl}, and they allow a simplification of theorem~\ref{channel_capacity_theorem}~\cite{note2}.

\begin{corollary}
With arbitrary qudit state as input and arbitrary encoding, and resources constrained by ${\bf Q}\left(\rho_{SW}\right)\ge{\bf y}$, the classical capacity of a $d$-dimensional generalized covariant channel $\Psi$ is
\be
\chi_L\left(\Psi\right)=S\left(\Psi(I/d)\right)-\min_{\substack{\varepsilon,\rho_{SW},\\{\bf Q}\left(\rho_{SW}\right)\ge{\bf y}}} E_{\Phi_{\mbox{$\varepsilon$}}\otimes \mathcal{I}}\left[\rho_{SW}\right].
\label{capacity_covariant_finite}
\ee
It is additive when the constraint has a separable form on each channel use and the encoding is LOCC.
\label{channel_capacity_theorem_covariant_finite}
\end{corollary}
Note that the encoding being considered is $\varepsilon$ plus unitaries in $G_U\left(\Psi\right)$.
Lower bounds of $\chi_L\left(\Psi\right)$ are obtained by choosing special $\varepsilon$; if $\varepsilon=\mathcal{I}$ (unitary encoding), $\Phi_{\mbox{$\varepsilon$}}$ is $\Psi$'s complementary channel $\Psi^c$ and we recover $\chi_L^{\mathcal{I}}\left(\Psi\right)$; if $\varepsilon=\mathcal{R}$, the map from all states to a pure state inside $\mathcal{H}_S$, we recover $\mathcal{C}^{(1)}$. Note here we do not require phase flips to guarantee achievability.

For the QEC~\cite{QEC}, Eqn.~\ref{capacity_covariant_finite} can be further simplified to
$
\chi_L\left(\Psi\right)= \max_{\varepsilon,\rho_{SW}}\left(1-\epsilon\right)\left(\log_2 d- E_{\varepsilon^c\otimes \mathcal{I}}\left[\rho_{SW}\right]\right),
$
where $\epsilon$ is the erasure probability~\cite{note2}. Let the quantum mutual information be the bipartite correlation measure in $Q\left(\rho_{SW}\right)\ge  y$. One can further obtain the lower bound~\cite{note2}
$
\chi_L^{\mathcal{I}}=\mathcal{C}_E\left(1-y/\left(2\log_2 d\right)\right),
$
where $\mathcal{C}_E=\left(1-\epsilon\right)2\log_2 d$~\cite{Wilde_book}. The other lower bound is  $\mathcal{C}^{(1)}=\mathcal{C}=\mathcal{C}_E/2$~\cite{Bennett_1999}. We observe that: at $y=2\log_2 d$, $\rho_{SW}$ is maximally entangled thus $\rho_S=I/d$, $\chi_L=\mathcal{C}^{(1)}$ while $\chi_L^{\mathcal{I}}=0$; at $y=0$, $\chi_L=\chi_L^{\mathcal{I}}=\mathcal{C}_E$. These two points are generic for all channels; when $0<y<2\log_2 d$, it is open what $\varepsilon$ allows $\chi_L\left(\Psi\right)$ to exceed $\max\left[\chi_L^{\mathcal{I}}, \mathcal{C}^{(1)}\right]$. Numerical results of quantum depolarizing channel~\cite{Chuang_book} suggest similar scaling behaviour with $y$~\cite{note2}.

{\em Application in quantum cryptography.---} 
We apply theorem~\ref{channel_capacity_theorem} in TW-QKD protocols to bound the general eavesdropper Eve's (coherent attack) information gain. Fig.~\ref{scheme_QKD} shows a general TW-QKD protocol~\cite{two_way_no_loss}. First, party-1 prepares a pure signal-reference pair $(R,W)$. Reference $W$ is kept by party-1 and a portion of it is used for security checking~\cite{note}. Then the signal $R$ goes through the forward channel controlled by Eve to party-2. Eve performs a unitary operation on $R$ and the pure mode $V$, producing her ancilla $E$ and $S$ for party-2. Note that in multiple channel uses, Eve's unitary operation can act on all signals jointly. Upon receiving $S$, party-2 uses a portion of the $S$ for security checking~\cite{note} and encodes a secret key on the rest of $S$ by a chosen scheme $\left(P_X\left(\cdot\right), ~\varepsilon_\cdot\right)$. The security checking by party-1 and party-2 jointly measures some functions ${\bf Q}\left(\rho_{SW}\right)$ of the state $\rho_{SW}$. Then the encoded signal goes through channel $\Psi$ in party-2~(\textit{e.g.}, device loss, amplification), leading to the output mode $B$, which is sent back to party-1 through the backward channel controlled by Eve. Finally, party-1 makes a measurement on the received mode and reference $W$ to obtain the secret key. 
\begin{figure}
\includegraphics[width=0.35\textwidth]
{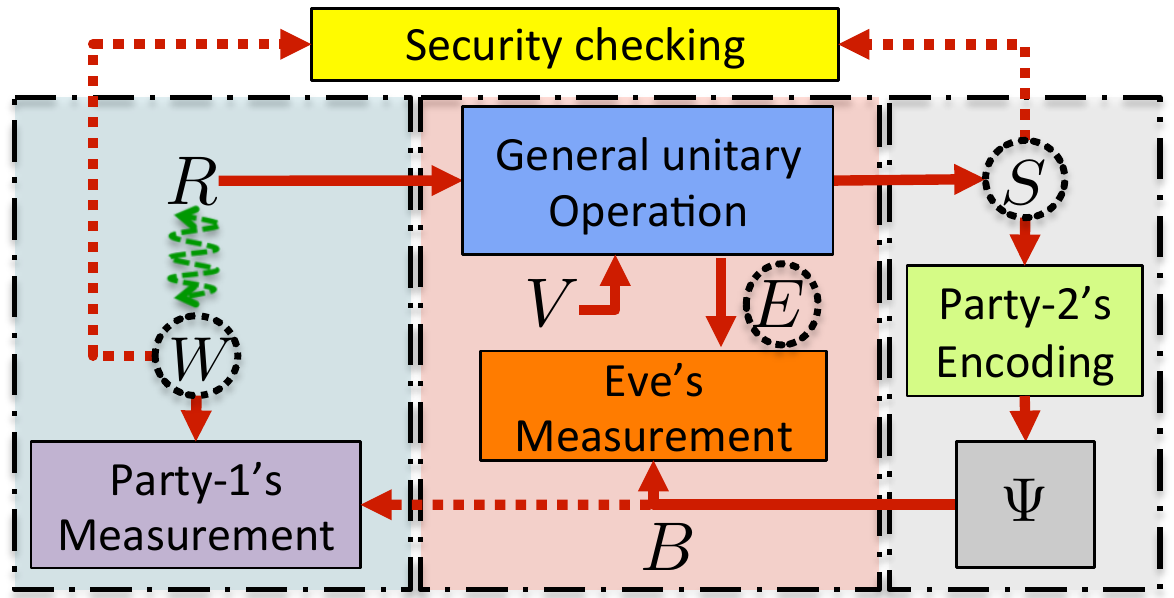}
\caption{Schematic of two-way QKD. The dotted circles highlight the three modes in the resource distribution step.}
\label{scheme_QKD}
\end{figure}
\begin{corollary}
In the TW-QKD protocol given above, the information gain per channel use of the eavesdropper's coherent attack is upper bounded by
$
\chi_L\left(\Psi\right)=\max_{\rho_{SW}}
F\left[\rho_{SW},\left(\,P_X\left(\cdot\right),\varepsilon_\cdot\,\right)\right],
$
where $F\left[\cdot\right]$ is defined in Eqn.~\ref{upperbound1},and the maximization is constrained by security checking measurement result ${\bf Q}\left(\rho_{SW}\right)={\bf y}$ and $\rho_W$ fixed.
\label{capacity_Eve}
\end{corollary}
\begin{proof}
To upper bound Eve's information gain, we give Eve all of $B$. This concession to Eve will not substantially increase Eve's information gain in long distance QKD, since the return fiber loss $\ll1$ (\textit{e.g.}, $\sim0.01$ at $100$ kilometers), which means almost all the light is leaked to Eve. Eve makes an optimal measurement on all $\left(B,E\right)$ pairs in multiple channel uses.

In a single run of the QKD protocol, $(S, E, W)$ is pure after Eve's unitary operation, the same as the scenario for theorem~\ref{channel_capacity_theorem}. Here $W$ is the witness---kept locally by party-1 and inaccessible to Eve; $E$ provides the resource as the pre-shared ancilla. The multiple QKD protocol runs also fit in our scenario. Moreover, party-1 and party-2 perform security checking to obtain constraints in the form of Eqn.~\ref{singleform} and Eqn.~\ref{Mform2} on $\rho_{SW}$. Controlled by party-2, the encoding operations are always LOCC. Eqn.~\ref{capacity_full} upper bounds the information gain per channel use of Eve's coherent attack.
\end{proof}

{\em Special case: Gaussian protocol.---}
If party-2 chooses the Gaussian channel $\Psi$ covariant with the unitary encoding, similar to corollary~\ref{channel_capacity_theorem_covariant_finite}, $\chi_L\left(\Psi\right)$ in corollary~\ref{capacity_Eve} has
\be
F\left[\rho_{SW},\left(\,P_X\left(\cdot\right),\varepsilon_\cdot\,\right)\right]=
S\left(\rho_B\right)-E_{\Psi^c\otimes \mathcal{I}}\left[\rho_{SW}\right].
\label{capacity_covariant}
\ee
For Gaussian protocols, the source $(R,W)$ and the channel $\Psi$ are Gaussian. The security checking functions are the mean photon number of $S$, and the cross-correlation between $S$ and $W$---both are functions of the covariance matrix $\Lambda_{SW}$ of $\rho_{SW}$. As a simplified form of Eqn.~\ref{upperbound1}, Eqn.~\ref{capacity_covariant} is subadditive. Moreover, $W$ is Gaussian and passive symplectic transforms~\cite{Weedbrook_2012} over $S$ preserve Eqn.~\ref{capacity_covariant}~\cite{Quntao_2015}, so the Gaussian extremality theorem~\cite{WolF_2006} applies. With all constraints on $\Lambda_{SW}$, Eqn.~\ref{capacity_covariant} is maximum when $\rho_{SW}$ is Gaussian. Thus for Gaussian protocols, the collective Gaussian attack is the most powerful.

{\em Discussion.---}
In future work, constraints in expectation value forms, \textit{i.e.} $\mathbb{E} \left[Q_k\left(\rho_{SW}\right)\right]\ge y_k$, extension of corollary.~\ref{channel_capacity_theorem_covariant_finite} to infinite dimensional systems and explicit evaluation of the capacity of QEC are of interest.

\begin{acknowledgements}
QZ is supported by the Claude E. Shannon Research Assistantship and AFOSR Grant No.~FA9550-14-1-0052. EYZ is supported by the National Science Foundation under grant Contract Number CCF-1525130. PWS is supported by the National Science Foundation under
grant Contract Number CCF-1525130, and by the NSF through
the STC for Science of Information under grant number CCF0-939370. The authors thank Zheshen Zhang, Jeffrey Shapiro, Aram Harrow and Zi-Wen Liu for helpful discussion. The authors also thank the anonymous referees for helpful feedback.
\end{acknowledgements}


\pagebreak
\begin{center}
\begin{widetext}
\textbf{\large Supplemental Materials: Additive Classical Capacity of Quantum Channels Assisted by Noisy Entanglement}
\end{widetext}
\end{center}
\setcounter{equation}{0}
\setcounter{figure}{0}
\setcounter{table}{0}
\setcounter{page}{1}
\makeatletter
\renewcommand{\theequation}{S\arabic{equation}}
\renewcommand{\thefigure}{S\arabic{figure}}
\renewcommand{\bibnumfmt}[1]{[S#1]}
\renewcommand{\citenumfont}[1]{S#1}


\section{Proof of theorem~1 in the main paper}
\subsection{Achievability of the upper bound}
\label{achievability}
The constrained capacity $\chi_L\left(\Psi\right)$ and the upper bound $\chi_L^{\rm UB}\left(\Psi\right)$ in Eqn.~\ref{upperbounddef} in the main paper are equal when the the unconditional state with optimum encoding and initial state choice satisfies the separability condition  
\be
\rho_{BE}=\rho_{B}\otimes\rho_{E}.
\ee
A special case where this is trivially satisfied is when there is no entanglement assistance. In that case, there is no ancilla $E$ and thus separability holds.

In general, the way to satisfy separability asymptotically is by the disentangling random phase flip encoding. After the encoding operation, $\left(S',E,W,C'\right)$ is pure. $S'$ and $\left(E,W,C'\right)$ can then be disentangled by performing the extra encoding in Ref.~\onlinecite{Bennett2002_sup} and Ref.~\onlinecite{Shor04_sup}, thus $I\left(S':EWC'\right)\simeq 0$. Here $I\left(A:B \right)=S\left(A\right)+S\left(B\right)-S\left(AB\right)$ is the quantum mutual information. Because $I\left(S':E\right)\le I\left(S':EWC'\right)\simeq 0$, we expect $S'$ and $E$ to be also in a product state after the disentangling. The formal proof is by block encoding~\cite{Shor04_sup}. 

Suppose the maximization $\chi_L^{\rm UB}\left(\Psi\right)$ in Eqn.~\ref{upperbounddef} in the main paper is acheived when the conditional states are $\sigma_x=\rho^{(x)}_{S^\prime}=\varepsilon_x \left[\rho_{S}\right]$, each with probability $p\left(x\right)$ and the joint state of $\left(S',E,W,C'\right)$ is $\ket{\Phi_{\mbox{$\varepsilon_x$}}}$. Then we can always create a block encoding by adding additional phase flips, which do not change $\chi_L^{\rm UB}\left(\Psi\right)$ in Eqn.~\ref{upperbounddef} in the main paper but make sure that $\chi_L\left(\Psi\right)= \chi^{\rm UB}_L\left(\Psi\right)$ asymptotically. The block encoding is the same as the one used in Ref.~\onlinecite{Shor04_sup}.

Let's consider $n\gg1$ number of states $\left\{\,\rho_1,\cdots, \rho_n\,\right\}$ sampled from the possible states $\sigma_x$'s, each with probability $p\left(x\right)$. There are $n_x= np\left(x\right)$ of $\sigma_x$ states. 
For the state $\sigma_x$ of $S'$, let $\ket{v_{xy}}$'s be its eigenstates and $\lambda_{xy}$'s the corresponding eigenvalues, \textit{i.e.}, $\sigma_x\equiv\sum_{y=1}^d \lambda_{xy} \ket{v_{xy}}\bra{v_{xy}}$. Then the joint state of $\left(S',E,W,C'\right)$, when $S'$ is in $\sigma_x$, is given by $\ket{\Phi_{\mbox{$\varepsilon_x$}}}\equiv\sum_{y=1}^d \sqrt{\lambda_{xy}}\ket{v_{xy}}_{S}\ket{v_{xy}}_{E W C'}$, where $d$ is the dimension of $S$. (For continuous variable (CV) system, a proper cutoff of dimension will work.) Besides the encoding operations $\varepsilon_\cdot$ which we already implemented to obtain the $\sigma_x$'s, Alice applies a random phase flip by a local unitary operation 
\begin{equation}
F_{x p}=\frac{1}{d}\sum_{y=1}^d(-)^{f\left(p,y\right)}\ket{v_{xy}}_{S}\bra{v_{xy}},
\end{equation}
where $f\left(p,y\right)=0,1$ arranges the different sign choices for each eigenvector and $p\in\left[1,2^{d-1}\right]$ is an index determining which of the $2^{d-1}$ phase changes is performed. We also consider the permutation of all $n$ states so that the states we encode with are symmetric. Consequently we want to use the $n!2^{n\left(d-1\right)}$ states of $\left(S^\prime,E\right)$'s to encode our message. The state of each input after the quantum channel $\Psi$ can be written as
\be
\rho_{\pi,P}^{\left(n\right)}=\otimes_{a=1}^n {\rm Tr}_{W C'}\left(\Psi\otimes \mathcal{I}\right)F_{\pi\left(a\right) P_{\pi\left(a\right)}}\left[\ket{\phi_{\pi\left(a\right)}}\bra{\phi_{\pi\left(a\right)}}\right].
\ee
The state is indexed by $\pi$ and $P$:
$\pi$ is one of the $n!$ permutations and $\pi(a)$ gives the index after the permutation at the location $a$; $P$ is a vector of length $n$, with each component $P_m\in\left[1,2^{d-1}\right], m\in[1,n]$ denoting the choice of the phase flips. Here we have used the notation that $U\left[\rho\right]\equiv U \rho U^\dagger$. $ \mathcal{I} $ acts on $ \left(E,W,C^\prime\right) $. The trace is over $\left(W,C^\prime\right)$ such that the state is for $\left(B, E\right)$.

We use HSW theorem~\cite{HSW1_sup,HSW2_sup} for the capacity $n\chi_L$ of $n$ uses of the channel, sending each $\rho_{\pi,P}^{\left(n\right)}$ with equal probability $1/n!2^{n\left(d-1\right)}$,
\be
n\chi_L=S\left(\frac{1}{n!2^{n\left(d-1\right)}}\sum_{\pi,P} \rho_{\pi,P}^{\left(n\right)}\right)-\frac{1}{n!2^{n\left(d-1\right)}}\sum_{\pi,P} S\left(\rho_{\pi,P}^{\left(n\right)}\right)
\label{chin}
\ee
The conditional term is simply $n$ times the original one in a single channel use, because $F_{xp}$ does not change the reduced density matrix $\sigma_x$ of $S^\prime$, thus commute with $ \Psi $ and the entropy of a product state is the sum of each reduced density matrix, \textit{i.e.},
\ba
&&S\left(\rho_{\pi,P}^{\left(n\right)}\right)\nonumber\\
&=&S\left(\otimes_{a=1}^n {\rm Tr}_{W C'}\left(\Psi\otimes \mathcal{I}\right)F_{\pi\left(a\right) P_{\pi\left(a\right)}}\left[\ket{\phi_{\pi\left(a\right)}}\bra{\phi_{\pi\left(a\right)}}\right]\right)\nonumber\\
&=&\sum_{a=1}^n S\left({\rm Tr}_{W C'}F_{\pi\left(a\right) P_{\pi\left(a\right)}}\left(\Psi\otimes \mathcal{I}\right)\left[\ket{\phi_{\pi\left(a\right)}}\bra{\phi_{\pi\left(a\right)}}\right]\right)\nonumber\\
&=&\sum_{a=1}^n S\left({\rm Tr}_{W C'}\left(\Psi\otimes \mathcal{I}\right)\left[\ket{\phi_{\pi\left(a\right)}}\bra{\phi_{\pi\left(a\right)}}\right]\right)\nonumber\\
&=&n\sum_x (\frac{n_x}{n}) S\left({\rm Tr}_{W C'}\left(\Psi\otimes \mathcal{I}\right)\left[\ket{\Phi_{\mbox{$\varepsilon_x$}}}\bra{\Phi_{\mbox{$\varepsilon_x$}}}\right]\right)\nonumber\\
&=&n\sum_x p\left(x\right)S\left(\rho_{N^\prime C^\prime W}^{\left(x\right)}\right).
\label{conditionaln}
\ea
Here we have used $F_{\pi\left(a\right) P_{\pi\left(a\right)}}$ is unitary and doesn't change the von Neumann entropy; $S\left({\rm Tr}_{W C'}\left(\Psi\otimes \mathcal{I}\right)\left[\ket{\Phi_{\mbox{$\varepsilon_x$}}}\bra{\Phi_{\mbox{$\varepsilon_x$}}}\right]\right)=S\left({\rm Tr}_{WC^\prime}\rho_{BE W C^\prime}^{\left(x\right)}\right)=S\left(\rho_{BE}^{\left(x\right)}\right)$ and $\left(B, E, W, N^\prime,C^\prime\right)$ is pure conditioned on the encoded message $x$. In the last equality we have used $n_x=np\left(x\right)$, which is true asymptotically. Then the second term in Eqn.~\ref{chin} is simply 
\be
\frac{1}{n!2^{n\left(d-1\right)}}\sum_{\pi,P} S\left(\rho_{\pi,P}^{\left(n\right)}\right)=n\sum_x p\left(x\right) S\left(\rho_{N^\prime C^\prime W}^{(x)}\right).
\ee
For the first term, first we need to use the same argument as in Ref.~\onlinecite{Shor04_sup} to sum over the phase flips,
\ba
&&\frac{1}{2^{n\left(d-1\right)}}\sum_{P} \rho_{\pi,P}^{\left(n\right)}\nonumber\\
&=&\frac{1}{2^{n\left(d-1\right)}}\otimes_{a=1}^n \sum_{P_{\pi(a)}}{\rm Tr}_{W C^\prime}\left(\Psi\otimes\mathcal{I}\right)
\nonumber
\\
&&\ \ \ \ \ \ \ \ \ \ \ \ \ \ \circ F_{\pi\left(a\right) P_{\pi\left(a\right)}}\left[\ket{\phi_{\pi\left(a\right)}}\bra{\phi_{\pi\left(a\right)}}\right]\nonumber\\
&=&\otimes_{a=1}^n  \sum_y {\lambda_{\pi\left(a\right)y}}\Psi\left[\ket{v_{\pi\left(a\right)y}}_{S}\bra{v_{\pi\left(a\right)y}}\right]
\nonumber\\
&&\ \ \ \ \ \ \ \ \otimes {\rm Tr}_{W C^\prime}\ket{v_{\pi\left(a\right)y}}_{E W C^\prime} \bra{v_{\pi\left(a\right)y}}\equiv \rho_{\pi}^{\left(n\right)}.
\nonumber
\ea
Now we need to further sum over the permutation. Since $n$ is large, by understanding the eigenvalue as a probability, and introducing another classical variable to denote which of eigenvalue (the same argument as Eqn.~8 in Ref.~\onlinecite{Shor04_sup}), asymptotically we have 
\begin{align}
&S\left(\frac{1}{n!}\sum_{\pi} \rho_{\pi}^{\left(n\right)}\right)&\nonumber\\
&=S\left(\frac{1}{n!}\sum_{\pi}  \otimes_{a=1}^n  \sum_y {\lambda_{\pi\left(a\right)y}}\Psi\left[\ket{v_{\pi\left(a\right)y}}_{S}\bra{v_{\pi\left(a\right)y}}\right]\right)&\nonumber\\
&+S\left(\frac{1}{n!}\sum_{\pi}\otimes_{a=1}^n \sum_y {\lambda_{\pi\left(a\right)y}}
{\rm Tr}_{WC^\prime}\ket{v_{\pi\left(a\right)y}}_{E W C^\prime} \bra{v_{\pi\left(a\right)y}}\right)&
\end{align}
Note that $\sum_y {\lambda_{\pi\left(a\right)y}}{\rm Tr}_{WC^\prime}\ket{v_{\pi\left(a\right)y}}_{EWC^\prime} \bra{v_{\pi\left(a\right)y}}=\rho_{E}$ does not depend on the permutation. So the second term above simply equals $nS\left(\rho_{E}\right)=nS\left(\rho_{SW}\right)$.
By Lemma.1 in Ref.~\onlinecite{Shor04_sup}, the first term asymptotically equals $nS\left(\bar{\rho}\right)$, where $\bar{\rho}$ is the average density matrix of $S^\prime$ after the quantum channel $\Psi$, which is exactly $\sum_x p\left(x\right) \Psi\left[\sigma_x\right]=\rho_{B}$. Thus asymptotically,
\be
S\left(\frac{1}{n!}\sum_{\pi} \rho_{\pi}^{\left(n\right)}\right)=n\left(S\left(\rho_{SW}\right)+S\left(\rho_{B}\right)\right)
\label{unconditionaln}
\ee

Combining the conditional term Eqn.~\ref{conditionaln} and the unconditional term Eqn.~\ref{unconditionaln}, we conclude that the information per channel use equals $\chi_L^{\rm UB}\left(\Psi\right)$ in Eqn.~\ref{upperbounddef} in the main paper. Thus we have proved that the upper bound is achievable.

\subsection{Additivity of the upper bound}
\label{additivity}


In this section, we want to show that the upper bound is additive.  For now we consider the bipartite constraint~\ref{Mform2} in the main paper, thus the constraints are $V^{\left(M\right)}=\left\{\,V_m, m\in[1,M]\,\right\}$, where
$
V_m=\left\{\,\left(P_{X_m}\left(\cdot\right), \varepsilon_{\cdot}\right)\in \mathbb{G}, \rho_{S_m}\in \mathcal{B}\left(\mathcal{H}_S\right),
{\bf Q}\left(\rho_{S_m,W_m}\right)\ge  {\bf y}\,\right\}
$
is the constraint on $m$th channel use.
Note that here we only put constraint on the marginal distribution of the encoding in each channel use. In general, we allow the encoding between different channel uses to be arbitrary LOCC, \textit{i.e.}, they can be classically correlated, satisfying some joint distribution $P_{\bf B}\left(\cdot\right)$. Conditioning on the message ${\bf x}\equiv\left(x_1,\cdots, x_M\right)$, the encoding operation is a tensor product of operations $\otimes_{m=1}^M \varepsilon_{x_m}$. We write the upper bound for Holevo information for $M$ channel uses with LOCC encoding as 
\begin{align}
&\chi_L^{\rm UB}\left(\Psi^{\otimes M}\right)=&
\nonumber
\\
&\max_{V^{\left(M\right)}}
 \left\{\,
S\left(\rho_{\bf B}\right)-\sum_{\bf x} P_{\bf X}\left({\bf x}\right) E_{\left(\otimes_{m=1}^M \Phi^{\left(x_m\right)}\otimes\mathcal{I}\right)}\left[\rho_{\bf S,\bf W}\right]\,\right\},&
\label{upperboundM}
\end{align}
where the CPTP map $\Phi^{\left(x\right)}$ is defined as the map from $\rho_{S}$ to $\rho_{N^\prime C^\prime}^{\left(x\right)}$.
First, from subadditivitry of von-Neumann entropy we have
\be
S\left(\rho_{\bf B}\right)\le \sum_{m=1}^M S\left(\rho_{B_m}\right),
\label{entropy_sub}
\ee
with equality when $\rho_{\bf B}$ is in a product state.
By Eqn.~\ref{superadditivty} in theorem~\ref{super_additivity_theorem} proven later, we obtain
\be
E_{\left(\otimes_{m=1}^M \Phi_m^{\left(x_m\right)}\right)\otimes I}\left[\rho_{\bf S,\bf W}\right] \ge \sum_{m=1}^M E_{\Phi_m^{\left(x_m\right)}\otimes\mathcal{I}}\left[\rho_{S_m,W_m}\right],
\label{entropy_gain_sup}
\ee
with equality when  $\left(S_m, W_m\right)$'s are conditionally in a product state for different $m$.
Combining Eqn.~\ref{entropy_sub} and Eqn.~\ref{entropy_gain_sup} and consider the marginal probability $P_{X_m}\left(x_m\right)=\sum_{{\bf x}\backslash x_m} P_{\bf X}\left({\bf x}\right)$, we have an upper bound for Eqn.~\ref{upperboundM} as follows
\be
\chi_L^{\rm UB}\left(\Psi^{\otimes M}\right) \le \max_{V^{\left(M\right)}} \sum_{m=1}^M F\left[\rho_{S_m,W_m},\left(\,P_{X_m}\left(\cdot\right),\varepsilon_\cdot\,\right)\right],
\ee
where $F$ is defined in Eqn.~\ref{upperbound1} in the main paper.
Due to the special form of the constraints $V^{\left(M\right)}$, which can be written as identical separate constraints on each term involved in the above summation, we obtain the upper bound as separate maximization on each channel use:
\be
\chi_L^{\rm UB}\left(\Psi^{\otimes M}\right) \le\sum_{m=1}^M \max_{V_m}F\left[\rho_{S_m,W_m},\left(\,P_{X_m}\left(\cdot\right),\varepsilon_\cdot\,\right)\right].
\label{bound}
\ee
The maximum can be achieved when $S_m$'s are in a product state with each other and are conditionally independent given $\bf W$. Note that this maximum has M terms, each term with identical constraint and expression with Eqn.~\ref{upperbound1} in the main paper, consequently we have $\chi_L^{\rm UB}\left(\Psi^{\otimes M}\right)\le M \chi_L^{\rm UB}\left(\Psi\right)$. While we must have $\chi_L^{\rm UB}\left(\Psi^{\otimes M}\right)\ge M \chi_L^{\rm UB}\left(\Psi\right)$, which is given by independent use as a lower bound. Consequently we obtain the additivity of the upper bound for Eqn.~\ref{upperbound1} in the main paper given the constraints $V$ and the $M$-channel-use constraints $V^{\left(M\right)}$.
\subsection{Alternative upper bound}
\label{alternative_UB}
If we consider the alternative constraint right above~\ref{Mform2} in the main paper (which is also given here in Eqn.~\ref{Mform2_supp}) and define $V^{\left(M\right)\prime }=\left\{\,V_m', m\in[1,M]\,\right\}$, where
$
V_m'=\left\{\,\left(P_{X_m}\left(\cdot\right),\varepsilon_\cdot\right)\in \mathbb{G}, \rho_{S_m}\in \mathcal{B}\left(\mathcal{H}_S\right),\right.
\left.{\bf Q}\left(\rho_{S_m,\bf W}\right)\ge {\bf y}\right\}
$ is the constraint on $m$th channel use.
By Eqn.~\ref{superadditivty2} in the superadditivity theorem~\ref{super_additivity_theorem2} proven later, we obtain
\be
E_{\left(\otimes_{m=1}^M \Phi_m^{\left(x_m\right)}\right)\otimes\mathcal{I}}\left[\rho_{\bf S,\bf W}\right] \ge \sum_{m=1}^M E_{\Phi_m^{\left(x_m\right)}\otimes\mathcal{I}}\left[\rho_{S_m,\bf W}\right],
\label{entropy_gain_sup2}
\ee
with equality when the $S_m$'s are conditionally independent given $\bf W$.
Combining Eqn.~\ref{entropy_sub} and Eqn.~\ref{entropy_gain_sup2} and consider the marginal probability $P_{X_m}\left(x_m\right)=\sum_{{\bf x}\backslash x_m} P_{\bf X}\left({\bf x}\right)$, we have an upper bound for Eqn.~\ref{upperboundM} as follows
\begin{align}
&\chi_L^{\rm UB}\left(\Psi^{\otimes M}\right) \le \max_{V^{\left(M\right)\prime}} \sum_{m=1}^M F\left[\rho_{S_m,\bf W},\left(\,P_{X_m}\left(\cdot\right),\varepsilon_\cdot\,\right)\right],&
\nonumber
\\
&\le  \sum_{m=1}^M \max_{V_m^\prime}  F\left[\rho_{S_m,\bf W},\left(\,P_{X_m}\left(\cdot\right),\varepsilon_\cdot\,\right)\right]&
\label{alter_bound}
\end{align}
where $F$ is defined in Eqn.~\ref{upperbound1} in the main paper with slight modification of the identity channel's dimension.

\subsection{Equivelence of two constraints}
In the main paper, we have two constraints (introduced around Eqn.~\ref{Mform2} in the main paper). The first one is
\be
{\bf Q}\left(\rho_{S_m\bf W}\right)\ge  {\bf y}, m\in[1,M].
\label{Mform1_supp}
\ee 
The second one is
\be
{\bf Q}\left(\rho_{S_mW_m}\right)\ge  {\bf y}, m\in[1,M].
\label{Mform2_supp}
\ee
We show that constraints \ref{Mform1_supp},~\ref{Mform2_supp} both give the same classical capacity. When constraint~\ref{Mform1_supp} is adopted, we have the upper bound Eqn.~\ref{alter_bound}; when constraint Eqn.~\ref{Mform2_supp} is adopted, we have the upper bound Eqn.~\ref{bound}. 
If we replace the notation ${\bf W} \leftrightarrow W_m$ in the constraint $V_m$ in the maximization of Eqn.~\ref{bound} and the constraint $V_m^\prime$ in the maximization of Eqn.~\ref{alter_bound}, we find that Eqn.~\ref{bound} and Eqn.~\ref{alter_bound} are the same, thus giving the same capacity.

\section{Generalized covariant channels}
\begin{figure}
\subfigure{
\includegraphics[width=0.25\textwidth]{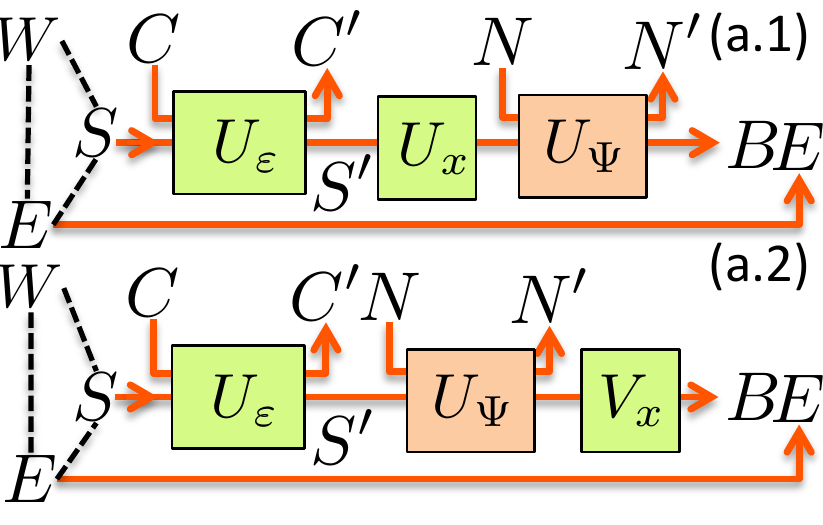}
\label{covariant_general_fig}
}
\subfigure{
\includegraphics[width=0.15\textwidth]{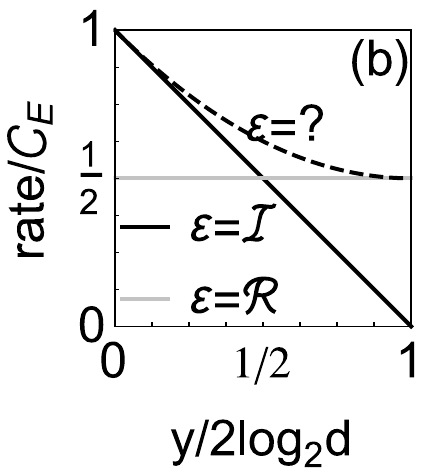}
\label{capacity_QEC_fig}
}
\caption{(a) Covariant channels. (b) capacities for QEC.
}
\end{figure}

To prove corollary~\ref{channel_capacity_theorem_covariant_finite} in the main paper, We first prove the following lemma.

\noindent
{\it {\bf Lemma.1}
For $d$ dimensional generalized covariant channel $\Psi$, $S(\Psi(\rho))\le S(\Psi(I/d))$ for all qudit state $\rho$.
}

\begin{proof}
For generalized covariant channel, consider the set $G_U\left(\Psi \right)$. For arbitrary qudit state $\rho$, we have
\begin{align}
S(\Psi(I/d))&=S\left(\sum_{U_x\in G_U\left(\Psi \right)}\frac{1}{d^2}\Psi\left(U_x\rho U_x^\dagger\right)\right)\nonumber\\
&\geq \frac{1}{d^2}\sum_{U_x\in G_U\left(\Psi \right)} S\left(\Psi\left(U_x\rho U_x^\dagger\right)\right)\nonumber\\
&=S\left(\Psi\left(\rho\right)\right).
\end{align}
\end{proof}

Now we prove corollary~\ref{channel_capacity_theorem_covariant_finite} in the main paper.

\begin{proof}
For input $\rho_{SW}$, let $\varepsilon^\star=\argmin_\varepsilon E_{\Phi_{\mbox{$\varepsilon$}}\otimes \mathcal{I}}\left[\rho_{SW}\right]$. $S\left(\rho_B\right)=S\left(\Psi\left(\sum_x P_X\left(x\right) \varepsilon_x \left[\rho_{S}\right]\right)\right)\le S\left(\Psi\left(I/d\right)\right)$ since $\Psi$ is generalized-covariant, due to lemma~1 proven above. Consequently, Eqn.~\ref{capacity_full} in the main paper is upper bounded by Eqn.~\ref{capacity_covariant_finite} in the main paper.

Now we show Eqn.~\ref{capacity_covariant_finite} in the main paper is achieved by the $\varepsilon^\star-$unitary encoding scheme, \textit{i.e.},   $\varepsilon^\star$ followed by $U_x\in G_U\left(\Psi\right)$ with equal probability (see Fig.~\ref{covariant_general_fig}). Because applying $U_x$'s with equal probability is equivalent to applying a fully depolarizing channel, which disentangles the input from other parties, the first term of Eqn.~\ref{HSW} in the main paper equals $ S\left(\rho_B\right)+S\left(\rho_E\right) $ and equality is obtained in Eqn.~\ref{upperbounddef} in the main paper. So from Eqn.~\ref{capacity_full} in the main paper, the capacity is given by 
\be
\chi_L\left(\Psi\right)=\max_{V}
F\left[\rho_{SW},\left(\,P_X\left(\cdot\right),\varepsilon_\cdot\,\right)\right],
\ee
with $F$ function in Eqn.~\ref{upperbound1} in the main paper simplified, due to the equivalence of Fig.~3(a.1),(a.2), to
\be
F\left[\rho_{SW},\left(\,P_X\left(\cdot\right),\varepsilon_\cdot\,\right)\right]=
S\left(\rho_B\right)-E_{\Phi_{\mbox{$\varepsilon^\star$}}\otimes \mathcal{I}}\left[\rho_{SW}\right],
\label{capacity_covariant_sup}
\ee
where $\rho_B=\Psi(\sum_{U_x\in G_U\left(\Psi\right)} U_x \varepsilon^\star\left[\rho_S\right]U_x^\dagger/d^2)=\Psi\left(I/d\right)$. Eqn.~\ref{capacity_covariant_finite} in the main paper is achieved after maximizing over $\rho_{SW}$.
\end{proof}

For covariant channels, the connection between our scenario and superdense coding on mixed states~\cite{non_unitary_optimum_sup} can be understood in Fig.~3(a.2). 
Note that Eqn.~\ref{capacity_covariant_sup} holds even if $\varepsilon^\star$ is not optimal, even for CV systems. However, corollary \ref{channel_capacity_theorem_covariant_finite} in the main paper only holds for DV systems.

Next, we give some examples of generalized covariant channels $\Psi$ and discuss $\chi_L\left(\Psi\right)$ given by corollary \ref{channel_capacity_theorem_covariant_finite} in the main paper and its lower bound $\chi_L^{\mathcal{I}}$ for unitary encodings.

\subsection{Quantum erasure channel}
Quantum erasure channel~(QEC)~\cite{QEC_sup} is the direct analog of the classical erasure channel. For input $\rho\in \mathcal{B}\left(\mathcal{H}\right)$, define an ``error'' state $\ket{e}\notin \mathcal{H}$, QEC is the CPTP map
\be
\Psi\left[\rho\right]=\left(1-\epsilon\right) \rho+\epsilon\ket{e}\bra{e},
\ee
where $\epsilon$ is the probability of erasure from the original state to the ``error'' state.
As
\be
\Psi\left[U\rho U^\dagger\right]=U\oplus I_e\left(\left(1-\epsilon\right) \rho+\epsilon\ket{e}\bra{e}\right)U^\dagger\oplus I_e,
\ee
QEC is covariant. We allow arbitrary input, \textit{i.e.}, $\mathcal{H}_S$ is the one qudit Hilbert space, and arbitrary encoding. 

The classical capacity given by Eqn.~\ref{capacity_covariant_finite} in the main paper can also be written explicitly as (schematic in Fig.~\ref{covariant_general_fig}2)
\be 
\chi_L\left(\Psi\right)= S\left(\Psi\left(I/d\right)\right)-\min_{\varepsilon,\rho_{SW}} \left(S\left(\rho_{N^\prime C^\prime W}\right)-S\left(\rho_{SW}\right)\right).
\label{covariant_general}
\ee 
The first term can be straightforwardly calculated as $S\left(\Psi\left(I/d\right)\right)=H_2\left(\epsilon\right)+\left(1-\epsilon\right)\log_2 d$,
where $H_2\left(\epsilon\right)=-\epsilon\log_2\epsilon-\left(1-\epsilon\right)\log_2\left(1-\epsilon\right)$ is the binary entropy function. Since $\rho_{N^\prime C^\prime W}=\Psi^c \otimes \mathcal{I}\left(\rho_{S^\prime C^\prime W}\right)$, where $\Psi^c$ is the complement of QEC, which is a QEC with parameter $1-\epsilon$. So 
\be
S\left(\rho_{N^\prime C^\prime W}\right)=H_2\left(1-\epsilon\right)+\epsilon S\left(\rho_{S^\prime C^\prime W}\right)+\left(1-\epsilon\right)S\left( \rho_{C^\prime W}\right)
.
\ee 
Note $H_2\left(1-\epsilon\right)=H_2\left(\epsilon\right)$ and $S\left(\rho_{S^\prime C^\prime W}\right)=S\left(\rho_{SW}\right)$, Eqn.~\ref{covariant_general} equals
\ba 
\chi_L\left(\Psi\right)
&=& \max_{\varepsilon,\rho_{SW}}\left(1-\epsilon\right)\left(\log_2 d- \left(S\left(\rho_{C^\prime W}\right)-S\left(\rho_{S W}\right)\right)\right)
\nonumber
\\
&=& \max_{\varepsilon,\rho_{SW}}\left(1-\epsilon\right)\left(\log_2 d- E_{\varepsilon^c\otimes \mathcal{I}}\left[\rho_{SW}\right]\right),
\label{covariant_general_QEC_sup}
\ea 
where the constraint is  ${\bf Q}\left(\rho_{SW}\right)\ge{\bf y}$. However, since the maximization is over both the input state and the quantum operation for state preparation, it is in general difficult to calculate the maximum.

We calculate the lower bound $\chi_L^{\mathcal{I}}\left(\Psi\right)$ by taking $\varepsilon=\mathcal{I}$,
\be 
\chi_L^{\mathcal{I}}\left(\Psi\right)= \max_{{\bf Q}\left(\rho_{SW}\right)\ge  {\bf y}}\left(1-\epsilon\right)\left(\log_2 d- S\left(\rho_W\right)+S\left(\rho_{SW}\right)\right).
\label{covariant_general_QEC_sup_2}
\ee 
If we choose $ {\bf Q} $ to be quantum mutual information, then under constraint $Q\left(\rho_{SW}\right)\ge y$,  Eqn.~\ref{covariant_general_QEC_sup_2} is maximum when $Q\left(\rho_{SW}\right)=y$ and $ \rho_{S}=I/d $. (For any $Q\left(\rho_{SW}\right)= y$, we can always choose the reduced density matrix to be $ \rho_{S}=I/d $ with Bell-diagonal states \cite{Bell_diagonal_sup}.) Finally, we have
\be
\chi_L^{\mathcal{I}}\left(\Psi\right)=\mathcal{C}_E\left(1-y/\left(2\log_2 d\right)\right),
\ee
where $\mathcal{C}_E=\left(1-\epsilon\right)2\log_2 d$ is the entanglement-assisted classical capacity for QEC~\cite{Wilde_book_sup}. When $ y=0 $, since $ S $ and $ E $ can be fully entangled, we recover the entanglement-assisted classical capacity.

For the QEC, as discussed in the main paper, here we plot $\chi_L^{\mathcal{I}}$ and $\mathcal{C}^{(1)}$ in Fig.~\ref{capacity_QEC_fig}: at $y=2\log_2 d$, $\rho_{SW}$ is maximally entangled thus $\rho_S=I/d$, $\chi_L=\mathcal{C}^{(1)}$ while $\chi_L^{\mathcal{I}}=0$; at $y=0$, $\chi_L=\chi_L^{\mathcal{I}}=\mathcal{C}_E$. These two points are generic for all channels, it is open what $\varepsilon$ allows $\chi_L\left(\Psi\right)$ to exceed $\max\left[\chi_L^{\mathcal{I}}, \mathcal{C}^{(1)}\right]$, as indicated by the dashed curve in Fig.~\ref{capacity_QEC_fig}. 

\subsection{Quantum depolarizing channel}
Quantum depolarizing channel (QDC)~\cite{Chuang_book_sup} is the direct analog to the classical binary symmetric channel. The $d$ dimensional QDC is defined by the CPTP map
\be
\Psi\left[\rho\right]=\frac{\lambda}{d}I+(1-\lambda)\rho.
\ee
As 
$
\Psi\left[U\rho U^\dagger\right]=U\Psi\left[\rho\right]U^\dagger,
$
$\Psi$ is covariant.

Thus we consider arbitrary input and $ Q\left(\rho_{SW}\right)\ge  y$ with $ Q $ the quantum mutual information.
The classical capacity $\chi_L$ given by Eqn.~\ref{capacity_covariant_finite} in the main paper is still difficult to obtain. So we take $\varepsilon=\mathcal{I}$ to obtain lower bound $\chi_L^{\mathcal{I}}$ by
\be 
\chi_L^{\mathcal{I}}\left(\Psi\right)= \max_{Q\left(\rho_{SW}\right)\ge  y}S\left(\Psi\left(I/d\right)\right)-E_{\Psi^{c}\otimes \mathcal{I}}\left[\rho_{SW}\right].
\label{Fun_new}
\ee 

The first term can be calculated easily $S\left(\Psi\left(I/d\right)\right)=\log_2d$.
The second term 
\be
E_{\Psi^{c}\otimes \mathcal{I}}\left[\rho_{SW}\right]=S\left( \Psi^{c}\otimes \mathcal{I}\left[\rho_{SW}\right]\right)-S\left(\rho_{SW}\right)
\ee
needs to be minimized under $Q\left(\rho_{SW}\right)\ge  y$.
We will use the fact that $Q\left(\rho_{SW}\right)$ is invariant under local unitaries and channel $\Psi^{c}$ is covariant~\cite{DP_complementary_sup,complementary_covariant_sup}. 

\begin{figure}
\includegraphics[width=0.3\textwidth]{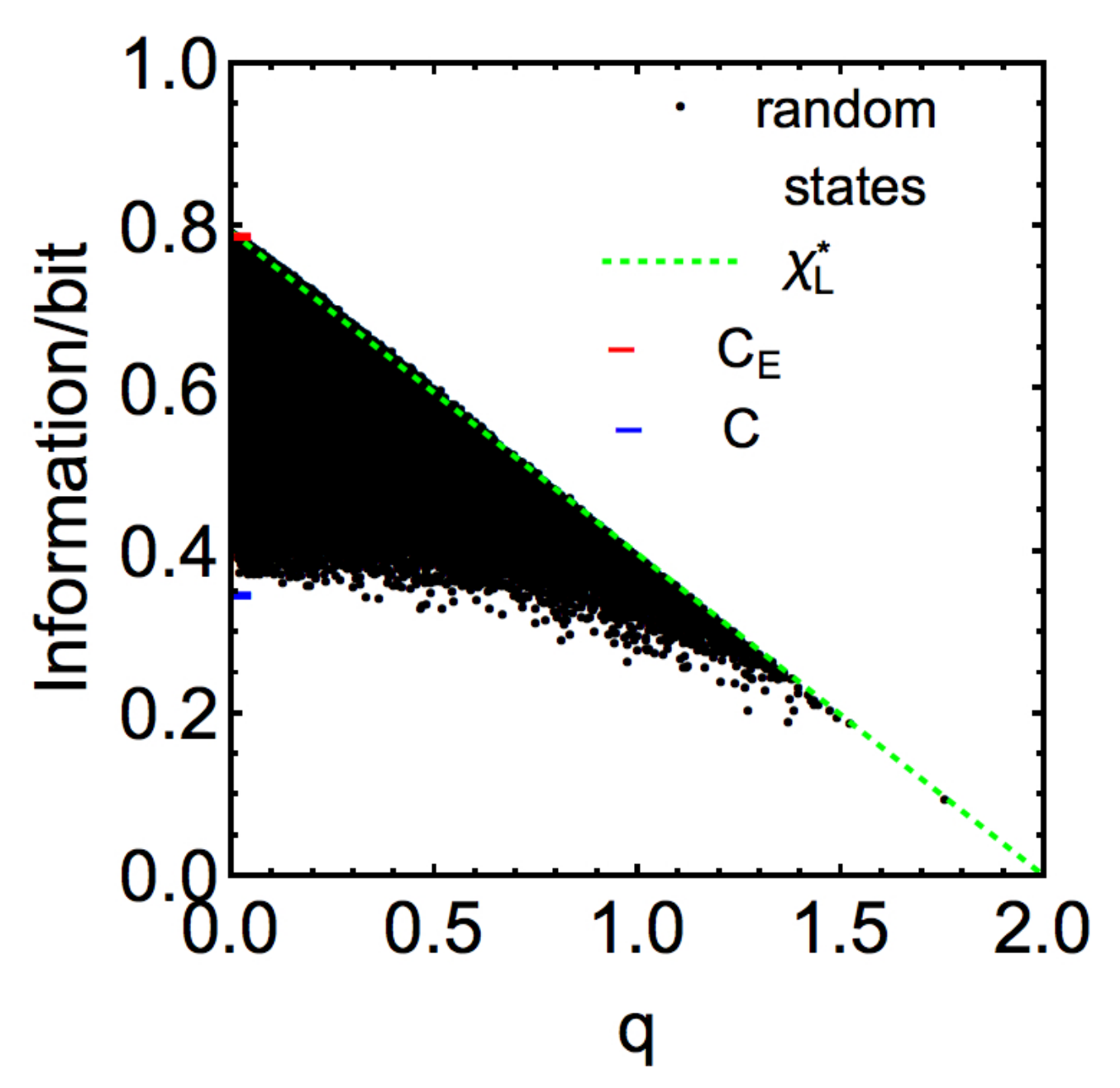}
\caption{ Qubit depolarizing channel. General random quantum state $F\left[\cdot\right]$ vs quantum mutual information $q$, sample size $10^6$.
\label{depolarizing_channela}
}
\end{figure}

Analytical result is difficult to obtain, here we consider the case of $d=2$ (\textit{i.e.}, single qubit) and obtain some numerical results. Without loss of generality, in the maximization of Eqn.~\ref{Fun_new} we consider the reduced density matrices $\rho_{S},\rho_{W}$ both diagonal. We can also apply local phase gate to choose the phase. As such, we can express the generic density matrix in the bases of $\ket{00}_{SW},\ket{01}_{SW},\ket{10}_{SW},\ket{11}_{SW}$ as
\be
\rho_{SW}=\left( 
\begin{array}{cccc}
a & c_3e^{-i\alpha}  & c_1e^{-i\beta} & r_1\\
c_3e^{i\alpha}  & p_2-a & r_2e^{i\theta} & -c_1e^{-i\beta} \\
c_1e^{i\beta} & r_2e^{-i\theta} & p_1-a & -c_3e^{-i\alpha} \\
r_1 & -c_1e^{i\beta} & -c_3e^{i\alpha} &1+a-p_1-p_2
\end{array} \right).
\ee 
where $\theta,\alpha,\beta\in[0,2\pi)$, $\max\left(0,p_1+p_2-1\right)\le a\le \left(p_1,p_2\right)\le 1$ and $1/2\ge r_1,r_2,c_1,c_3\ge 0$. There are ten variables $\left\{\theta,\alpha,\beta,c_1,c_3,r_1,r_2,p_1,p_2,a\right\}$. We randomly select those variables to generate random quantum states and calculate the quantum mutual information $q$ and $F\left[\cdot\right]$ function being maximized over in Eqn.~\ref{Fun_new}. To calculate $F\left[\cdot\right]$, we use the complementary channel $\Psi^{c}$ given in Ref.~\onlinecite{Leung_2015_sup} and explicitly obtain the output density matrix. The reuslts are in Fig.~\ref{depolarizing_channela}. The maximum at each $q$ fits well with $\chi_L^\star\equiv\mathcal{C}_E\left(1-q/\left(2\log_2 d\right)\right)$,
where $
\mathcal{C}_E=2+\left(1-\frac{3\lambda}{4}\right)\log_2\left(1-\frac{3\lambda}{4}\right)+\frac{3\lambda}{4}\log_2\frac{\lambda}{4}
$~\cite{Bennett_1999_sup}. We do notice a minute deviation (but significant compared with numerical precision) present in Fig.~\ref{depolarizing_channela} from $\chi_L^\star$, which will be studied in the future.

We also compare with its classical capacity~\cite{Bennett_1999_sup,King_2003_sup} 
$
\mathcal{C}=1+\left(1-\frac{\lambda}{2}\right)\log_2\left(1-\frac{\lambda}{2}\right)+\frac{\lambda}{2}\log_2\frac{\lambda}{2}.
$ We see no states below $\mathcal{C}$ for $q=0$. This means that entanglement always increases the amount of information that can be conveyed.

\subsection{Dephasing channel}
The qubit phase damping channel~(QPDC)~\cite{Chuang_book_sup} is defined by the CPTP map on the qubit density matrix 
\be
\Psi\left[\left( \begin{array}{cc}
p & b \\
b^\star & 1-p\end{array} \right)
\right]=\left( \begin{array}{cc}
p & \sqrt{1-\gamma}b \\
\sqrt{1-\gamma}b^\star & 1-p\end{array} \right).
\ee
One can show that the covariant group of QPDC are generated by phase gates ${Ph}\left(\cdot\right)$ and Pauli-X.  $\chi_L\left(\Psi\right)$ is given by corollary \ref{channel_capacity_theorem_covariant_finite} in the main paper, while further calculation will be performed in the future study.

\section{Entropy gain theorems}
First, let's review the quantum data processing inequality~\cite{quantum_data_processing_sup}. Here we also derive a conditional mutual information version of data processing inequality.
\begin{theorem}
{\rm (quantum data processing inequality.)}
For bipartite system $AB$, local quantum operations $\phi_A$ on sub-system $A$ and $\phi_B$ on sub-system $B$ we have
\begin{equation}
I\left(\phi_A\left[A\right]:\phi_B\left[B\right]\right)\le I\left(A:B\right).
\label{data_process1}
\end{equation}
For a composite system $ABC$, consider local quantum operations $\phi_A$ on sub-system A and $\phi_B$ on sub-system B. Suppose $C$ is unchanged, we can obtain a conditional version of the above inequality
\begin{equation}
I\left(\phi_A\left[A\right]:\phi_B\left[B\right]|C\right)\le I\left(A:B|C\right),
\label{data_process2}
\end{equation}

\label{data_process_theorem}
\end{theorem}
We will give a proof here. The nature of quantum processing inequality is subadditivity of von Neumann entropy. One can also think of subadditivity from the Uhlmann's theorem\cite{uhlmann1977_sup, Hayden2004_sup} that any quantum operation $T$ does not increase quantum relative entropy, \textit{i.e.}, 
\be
S\left(\rho\|\sigma\right)\ge S\left(T\left[\rho\right]\|T\left[\sigma\right]\right).
\ee
The mutual information can be expressed by the relative entropy as
\be
I\left(A:B\right)=S\left(\rho_{AB}\| \rho_{A}\otimes \rho_{B}\right).
\label{mutual_int}
\ee
Similarly, the conditional mutual information 
\begin{align}
&I\left(A:B|C\right)=S\left(AC\right)+S\left(BC\right)-S\left(ABC\right)-S\left(C\right)& \nonumber\\
&=S\left(\rho_{ABC}\|\rho_{A}\otimes \rho_{BC}\right)-S\left(\rho_{AC}\|\rho_{A}\otimes \rho_{C}\right)&
\label{mutual_int2}
\\
&=S\left(\rho_{ABC}\|\rho_{B}\otimes \rho_{AC}\right)-S\left(\rho_{BC}\|\rho_{B}\otimes \rho_{C}\right)&
\label{mutual_int3}
\end{align}
Thus consider the operation local on each party $T=\phi_A \otimes \phi_B$, from Eqn.~\ref{mutual_int} we have 
\begin{align}
&I\left(\phi_A\left[A\right]:\phi_B\left[B\right]\right)=S\left(T\left[\rho_{AB}\right]\| T\left[\rho_{A}\otimes \rho_{B}\right]\right)&\\
&\le S\left(\rho_{AB}\| \rho_{A}\otimes \rho_{B}\right)=I\left(A:B\right)&
\end{align}
 This proves Eqn.~\ref{data_process1}. Similarly, from Eqn.~\ref{mutual_int2} we have
\begin{align}
&I\left(\phi_A\left[A\right]:\phi_B\left[B\right]|C\right)=&\nonumber\\
&S\left(\phi_A\otimes\phi_B\otimes\mathcal{I}\left[\rho_{ABC}\right]\|\phi_A\otimes\phi_B\otimes\mathcal{I}\left[\rho_{A}\otimes \rho_{BC}\right]\right)&\nonumber\\
&\ \ \ \ \ -S\left(\phi_A\otimes\mathcal{I}\left[\rho_{AC}\right]\|\phi_A\left[\rho_{A}\right]\otimes \rho_{C}\right)&\\
&\le S\left(\phi_A\otimes\mathcal{I}\otimes\mathcal{I}\left[\rho_{ABC}\right]\|\phi_A\otimes\mathcal{I}\otimes\mathcal{I}\left[\rho_{A}\otimes \rho_{BC}\right]\right)&\nonumber\\
&\ \ \ \ \ -S\left(\phi_A\otimes\mathcal{I}\left[\rho_{AC}\right]\|\phi_A\left[\rho_{A}\right]\otimes \rho_{C}\right)&\\
&=I\left(\phi_A\left[A\right]:B|C\right).&
\end{align}
Repeat the above procedure for $I\left(\phi_A\left[A\right]:B|C\right)$ again, but based on Eqn.~\ref{mutual_int3}, as follows
\begin{align}
&I\left(\phi_A\left[A\right]:B|C\right)=&\nonumber\\
&S\left(\phi_A\otimes\mathcal{I}\otimes\mathcal{I}\left[\rho_{ABC}\right]\| \phi_A\otimes\mathcal{I}\otimes\mathcal{I}\left[\rho_{B}\otimes \rho_{AC}\right]\right)&\nonumber\\
&\ \ \ \ \ -S\left(\rho_{BC}\|\rho_{B}\otimes \rho_{C}\right)&\\
&\le S\left(\rho_{ABC}\|\rho_{B}\otimes \rho_{AC}\right)-S\left(\rho_{BC}\|\rho_{B}\otimes \rho_{C}\right)&\\
&=I\left(A:B|C\right).&
\end{align}
Consequently, combining the above we have proved Eqn.~\ref{data_process2}. Thus we have proved the quantum data processing inequality.

Based on theorem~\ref{data_process_theorem}, we can prove the superadditivity of entropy gain, as given by theorem~\ref{super_additivity_theorem}, \ref{super_additivity_theorem2}. We use the short notation $A_{\left(1:D\right)}$ for $A_1A_2\cdots A_D$.
\begin{theorem}
{\rm (Superadditivity of entropy gain I.) }
Given a $2D$-partite state $\rho_{A_{\left(1:D\right)}B_{\left(1:D\right)}}\in \mathcal{B}\left(\mathcal{H}_A^{\otimes D}\otimes \mathcal{H}_B^{\otimes D}\right)$ and $D$ CPTP maps $\phi_d:\mathcal{B}\left(\mathcal{H}_A\otimes \mathcal{H}_B\right)\to \mathcal{B}\left(\mathcal{H}_A'\otimes\mathcal{H}_B'\right), 1\le d\le D$. Consider the CPTP map $\otimes_{d=1}^D \phi_d: \mathcal{B}\left(\mathcal{H}_A^{\otimes D}\otimes \mathcal{H}_B^{\otimes D}\right)\to \mathcal{B}\left({\mathcal{H}_A'}^{\otimes D}\otimes {\mathcal{H}_B'}^{\otimes D}\right)$. The entropy gain is superadditive,
\begin{equation}
E_{\otimes_{d=1}^D \phi_d}\left[\rho_{A_{\left(1:D\right)}B_{\left(1:D\right)}}\right]\ge \sum_{d=1}^D E_{\phi_d} \left[\rho_{A_{d}B_{d}}\right],
\label{superadditivty}
\end{equation}
with equality achieved when $\rho_{A_{\left(1:D\right)}B_{\left(1:D\right)}}=\otimes_{k=1}^D \rho_{A_kB_k}$
\label{super_additivity_theorem}
\end{theorem}
\begin{theorem}
{\rm (Superadditivity of entropy gain II.) }

Given a $D$-partite state $\rho_{A_{\left(1:D\right)}R}\in \mathcal{B}\left(\mathcal{H}_A^{\otimes D}\otimes \mathcal{H}_R\right)$ and D CPTP maps $\phi_d: \mathcal{B}\left(\mathcal{H}_A\to \mathcal{H}_A^\prime\right), 1\le d\le D$. Consider the CPTP map $\otimes_{d=1}^D \phi_d: \mathcal{B}\left(\mathcal{H}_A^{\otimes D}\to {\mathcal{H}_A'}^{\otimes D}\right)$ The entropy gain is super-additive,
\begin{equation}
E_{\left(\otimes_{d=1}^D \phi_d\right)\otimes\mathcal{I}}\left[\rho_{A_{\left(1:D\right)}R}\right]\ge \sum_{d=1}^D E_{\phi_d \otimes\mathcal{I}} \left[\rho_{A_{d}R}\right],
\label{superadditivty2}
\end{equation}
with equality achieved when $\rho_{A_d}$'s are conditionally independent on $R$.
\label{super_additivity_theorem2}
\end{theorem}

\subsection{Proof of theorem~\ref{super_additivity_theorem}}
\label{proof1}
Denote $\phi^{\left(k+1\right)}\equiv\otimes_{d=1}^{k+1} \phi_d$ for convenience, we only need to show the inequality
\begin{eqnarray}
&E_{\phi^{\left(k+1\right)}}\left[\rho_{A_{\left(1:k+1\right)}B_{\left(1:k+1\right)}}\right]\nonumber\\
&\ge E_{\phi^{\left(k\right)}}\left[\rho_{A_{\left(1:k\right)}B_{\left(1:k\right)}}\right]+E_{\phi_{k+1}} \left[\rho_{A_{k+1}B_{k+1}}\right]
\label{inequalityk}.
\end{eqnarray}
Then by starting at $k=D-1$ and apply the inequality repeatedly until $k=1$ we will arrive at Eqn.~\ref{superadditivty}. By expansion from the definition of entropy gain in the main paper and noticing $\phi^{\left(k+1\right)}=\phi^{\left(k\right)} \otimes \phi_{k+1}$, Eqn.~\ref{inequalityk} is equivalent to
\begin{eqnarray}
&S\left(\phi^{\left(k+1\right)}\left[\rho_{A_{\left(1:k+1\right)}B_{\left(1:k+1\right)}}\right]\right)-S\left(\rho_{A_{\left(1:k+1\right)}B_{\left(1:k+1\right)}}\right)\nonumber\\
&\ge S\left(\phi^{\left(k\right)}\left[\rho_{A_{\left(1:k\right)}B_{\left(1:k\right)}}\right]\right)-S\left(\rho_{A_{\left(1:k\right)}B_{\left(1:k\right)}}\right)\nonumber\\
&+S\left(\phi_{k+1}\left[\rho_{A_{k+1}B_{k+1}}\right]\right)-S\left(\rho_{A_{k+1}B_{k+1}}\right).
\label{entropy_full}
\end{eqnarray}
It is easy to see it is equivalent to
\begin{align}
&
I\left(\phi^{\left(k\right)}\left[A_{\left(1:k\right)}B_{\left(1:k\right)}\right]:\phi_{k+1}\left[A_{k+1}B_{k+1}\right]\right)
&
\nonumber\\
&
\le I\left(A_{\left(1:k\right)}B_{\left(1:k\right)}:A_{k+1}B_{k+1}\right).
&
\label{mutualinfo}
\end{align} 
Eqn.~\ref{mutualinfo} is true due to quantum data processing inequality Eqn.~\ref{data_process1} because the CPTP map $\phi^{\left(k+1\right)}$ acts on local parties independently.

\subsection{Proof of theorem~\ref{super_additivity_theorem2}}
We only need to replace three equations in the proof of theorem~\ref{super_additivity_theorem}. First, we replace Eqn.~\ref{inequalityk} with
\begin{eqnarray}
&E_{\phi^{\left(k+1\right)}\otimes\mathcal{I}}\left[\rho_{A_{\left(1:k+1\right)}R}\right]\nonumber\\
&\ge E_{\phi^{\left(k\right)}\otimes\mathcal{I}}\left[\rho_{A_{\left(1:k\right)}R}\right]+E_{\phi_{k+1}\otimes\mathcal{I}} \left[\rho_{A_{k+1}R}\right]
\label{inequalityk2}.
\end{eqnarray}
Then, replace Eqn.~\ref{entropy_full} with
\begin{eqnarray}
&S\left(\phi^{\left(k+1\right)}\otimes\mathcal{I}\left[\rho_{A_{\left(1:k+1\right)}R}\right]\right)-S\left(\rho_{A_{\left(1:k+1\right)}R}\right)\nonumber\\
&\ge S\left(\phi^{\left(k\right)}\otimes\mathcal{I}\left[\rho_{A_{\left(1:k\right)}R}\right]\right)-S\left(\rho_{A_{\left(1:k\right)}R}\right)\nonumber\\
&+S\left(\phi_{k+1}\otimes\mathcal{I}\left[\rho_{A_{k+1}R}\right]\right)-S\left(\rho_{A_{k+1}R}\right).
\label{entropy_full2}
\end{eqnarray}
Finally, replace Eqn.~\ref{mutualinfo} with
\be
I\left(\phi^{\left(k\right)}\left[A_{\left(1:k\right)}\right]:\phi_{k+1}\left[A_{k+1}\right]|R\right)\le I\left(A_{\left(1:k\right)}:A_{k+1}|R\right).
\ee
And then for the same reason, it is true due to quantum data processing inequality Eqn.~\ref{data_process2} because the CPTP map $\phi^{\left(k+1\right)}$ doesn't change $R$ and acts on local parties independently.

\section{More discussions}
\label{discussion_Shor}

Here we point out the connection between our result and the result in Ref.~\onlinecite{Shor04_sup} by showing that a special case of Ref.~\onlinecite{Shor04_sup}'s result and a special case of our result is equivalent. For our case, we can choose the inaccessible $W_m$'s to be pure and the bipartite measure in the constraints to be effectively single-partite $Q(\rho_{SW})=-S({\rm Tr}_{W} \rho_{SW})$. Then $W_m$ is in product state with everything else and the constraints of Eqn.~\ref{singleform} in the main paper or Eqn.~\ref{Mform1_supp} become constraints on the input entropy, \textit{i.e.}, $S(\rho_{S})\le P$ for some $P>0$. We further restrict the encoding to be unitary operations, thus the encoded state is also constrained by $S(\rho_{S'}^{\left(x\right)})\le P$. Also we see that all the purification of $\rho_{S'}^{\left(x\right)}$ has been sent to Bob, just like in Ref.~\onlinecite{Shor04_sup}'s case. For Ref.~\onlinecite{Shor04_sup}'s result, we restrict the input signal entropy $S(\rho_{S}^{\left(x\right)})\le P$ in each channel use instead of the original expectation form in Ref.~\onlinecite{Shor04_sup} and require that different states indexed by message $x$ are unitary equivalent. After the above restrictions, both formula are equal. In other words, the capacity formula in Ref.~\onlinecite{Shor04_sup} with the above restriction is additive. 

We also want to point out that only in the scenario of two step channel use---resource distribution and encoding, the unconditional state $\rho_{\bf E}$ equals the conditional state $\rho_{\bf E}^{\left(x\right)}$. This allows the entropy gain form in Eqn.~\ref{upperbound1} in the main paper, which is crucial for proving additivity.


\begin{thebibliography}{50}
\bibitem{locc_explain}
Because Alice has all the information of her state, measurement is not necessary, thus it's adequate to consider CPTP maps.
\bibitem{note}
This is realized either by a beam-splitter or a switch
  between security checking and key generation.
\end{thebibliography}
\end{document}